\documentclass[letterpaper, 10pt, conference]{ieeeconf}      

\IEEEoverridecommandlockouts                              

\usepackage{graphics} 
\usepackage{epsfig} 
\usepackage{mathptmx} 
\usepackage{times} 
\usepackage{amsmath} 
\usepackage{amssymb}  
\usepackage{cite}
\usepackage{graphicx} 
\usepackage{tabularx}
\usepackage{hyperref}
\usepackage{xcolor}
\usepackage{multicol}
\usepackage{bbm} 
\usepackage{dsfont}

\newtheorem{thm}{Theorem}
\newtheorem{lem}{Lemma}
\newtheorem{prop}{Proposition}
\newtheorem{rem}{Remark}
\newtheorem{defi}{Definition}

\DeclareMathOperator*{\esssup}{ess \, sup}

\newcommand{\dx}{\, \mathrm d x}
\newcommand{\dy}{\, \mathrm d y}
\newcommand{\R}{\mathbb R}
\newcommand{\eps}{\varepsilon}

\newcommand{\vertiii}[1]{{\left\vert\kern-0.25ex\left\vert\kern-0.25ex\left\vert #1
    \right\vert\kern-0.25ex\right\vert\kern-0.25ex\right\vert}}

\title{\LARGE \bf Opinion dynamics on signed graphs and graphons:\\ Beyond the piece-wise constant case
}

\author{Raoul Prisant, Federica Garin, Paolo Frasca  
\thanks{This work has been partly supported by the French National Research Agency through grant COCOON ANR-22-CE48-0011.}
\thanks{All authors are with Univ.\ Grenoble Alpes, CNRS, Inria, Grenoble INP, GIPSA-lab, 38000 Grenoble, France (e-mails: paolo.frasca@gipsa-lab.fr; federica.garin@inria.fr; raoul.prisant@gipsa-lab.fr).} 
}

\renewcommand{\baselinestretch}{.98}

\begin{document}

\maketitle

\thispagestyle{empty}
\pagestyle{empty}

\begin{abstract}
    In this paper we make use of graphon theory to study opinion dynamics on large undirected networks. The opinion dynamics models that we take into consideration allow for negative interactions between the individuals, i.e.\ competing entities whose opinions can grow apart. We consider both the repelling model and the opposing model that are studied in the literature. We define the repelling and the opposing dynamics on graphons and we show that their initial value problem's solutions exist and are unique. We then show that the graphon dynamics well approximate the dynamics on large graphs that converge to a graphon. This result applies to large random graphs that are sampled according to a graphon. All these facts are illustrated in an extended numerical example. 
\end{abstract}


\section{Introduction}
For more than a decade, the control community has been interested in mathematical models of opinion dynamics on social networks. 
This body of research has been summarized by several survey papers~\cite{proskurnikov2017tutorial,proskurnikov2018tutorial,shi2019dynamics,tian2023dynamics,bernardo2024bounded}. A key question has been to determine whether or not in the long run the individuals in the social network reach a consensus, whereby their opinions are in agreement~\cite{friedkin2015problem}. A popular feature of non-agreement models is the presence of negative, or antagonistic, interactions between individuals~\cite{altafini2012consensus,shi2019dynamics}. Indeed, negative interactions counter the positive interactions that are commonplace in social influence models and that lead individuals to conform with their peers.

Opinion dynamics take place on social networks that can be very large: therefore, our mathematical methods need be able to cope with large networks. In this perspective of scalability, the control community has become interested, in the last few years, in methods that represent network dynamics by the evolution of continuous variables~\cite{gao2017control,gao2019graphon,nikitin2021continuation,maffettone2022continuification}. A useful instrument in this perspective is the theory of graph limits, and particularly the notion of graphon, which was developed in the 2000s and is thoroughly presented in~\cite{lovasz2012book}. Using graphons to study dynamical systems on networks is currently a very active area of research~\cite{gao2019graphon,avella2018centrality,vizuete2020graphon,belabbas2021h,petit2021random,bonnet2022consensus,chen2023large}.

In this paper we consider opinion dynamics on signed graphs and we define their graphon counterparts. More precisely, we consider both the so-called~\cite{shi2019dynamics} {\em repelling model} and the {\em opposing model}, two possible extensions of the classical French-DeGroot model that allow the agents to have negative interactions.
%
We give three contributions regarding these two dynamics on graphs and graphons. First, we prove the existence and uniqueness of the solutions to these dynamics on signed graphons (Theorem~\ref{thm:existence-and-unicity}). 
Second, we prove sufficient conditions for the solutions on signed graphs to converge, as the number of individuals $n$ goes to infinity, to solutions on signed graphons, as long as the sequence of graphs converges to a graphon (the convergence error is estimated in Theorem~\ref{thm:bound}). These sufficient conditions are general enough to apply to relevant cases, including sequences of random graphs sampled from piece-wise Lipschitz graphons (our third contribution, Theorem~\ref{thm:sampled-convergence}).

\paragraph*{Related work} 
Several papers have considered opinion dynamics on graphons, related Laplacian-based dynamics on graphons, or  the properties of graphon Laplacians.
The overwhelming majority of works assume
nonnegative interactions, including~\cite{vonluxburg2008consistency,vizuete2021laplacian,bramburger2023pattern,petit2021random}.
Papers~\cite{bonnet2022consensus,ayi2021mean} contains a study of graphon opinion dynamics with nonnegative, but possibly time-varying, interactions. Distributed optimization, which is closely related to consensus dynamics, is studied in~\cite{chen2023large}. Some works have also considered more general nonlinear dynamics on networks. For instance,  the well-known work by Medvedev~\cite{medvedev2014nonlinear,medvedev2014Wrandom}, on which much of our analysis is based, 
covers nonlinear dynamics such as the Kuramoto model of synchronization. In fact, our proof arguments for Theorems~\ref{thm:existence-and-unicity} and~\ref{thm:bound} are adaptations of Medvedev's~\cite{medvedev2014nonlinear}, in which we restrict the analysis to linear dynamics and extend it to signed graphons and sampled graphs.

Signed graphons have been sometimes considered in the literature~\cite{lovasz2011subgraph}, but rarely as a dynamical model, even though the interest of large-scale signed social networks has been recognized for quite some time~\cite{facchetti2011computing}.
Recently, the paper \cite{aletti2022opinion} has considered the repelling dynamics on signed graphs and has proved existence of solutions for piece-wise constant graphons. 
Our results are more general than the latter as they cover the opposing dynamics, include explicit approximation bounds, consider the case of random sampled graphs, and make weaker assumptions on the regularity of the graphons (no assumption for Theorem~\ref{thm:bound} and piece-wise Lipschitz continuity for Theorem~\ref{thm:sampled-convergence}). 

\paragraph*{Outline}
In Section~\ref{sect:opinion-dynamics}, we recall the relevant models of opinion dynamics on signed graph and we define their graphon counterparts. We then prove that the latter have complete classical solutions from any initial condition. In Section~\ref{sect:convergence-solutions}, we prove sufficient conditions for the solutions of the dynamics on graphs of size $n$ to converge, when $n\to\infty$, to solutions of the graphon dynamics. In Section~\ref{sect:convergence-sampled} we show that  our convergence conditions apply to  large random graphs sampled from signed graphons (so long as the graphon is Lipschitz). 
Finally, Section~\ref{sect:simulations} contains an illustrative example and Section~\ref{sect:conclusion} reports some concluding remarks.

\section{Opinion dynamics on graphs and graphons}\label{sect:opinion-dynamics}

In this section, we recall established models of opinion dynamics on signed graph and we define their counterparts on signed graphons. On our way to the latter, we recall the necessary preliminaries about graphons and about random graphs sampled from graphons.

\subsection{Opinion dynamics on signed networks}
Opinion dynamics study the evolution of the opinions of interacting individuals~\cite{proskurnikov2017tutorial,proskurnikov2018tutorial}. The interaction network is modeled by an undirected graph $G^{(n)}$, where the nodes $v_i$ represent the individuals and the edges $e_{ij}$ represent an interaction between $v_i$ and $v_j$. In general, a weight matrix $A^{(n)}$ 
can be assigned to determine the strength of each connection.
For the scope of this paper, we will take symmetric matrices $A^{(n)}$ with entries $A_{ij}^{(n)}\in\{-1,0,1\}$, where value $1$ represents a positive interaction, the two individuals respect each other and as such their opinions grow closer, value $0$ means lack of interaction and value $-1$ represents the interaction between individuals who dislike each other and whose opinions tend to diverge. 

Two relevant ways of defining an opinion dynamics associated to such signed graphs have been considered in the literature~\cite[Sect.~2.5.2]{shi2019dynamics}. 
The first is the so-called {\em repelling} model, in which the opinion of each node $i$ is modeled as $u^{(n)}_i\in\mathbb{R}$, and whose evolution in time is determined by the differential equation
\begin{equation}\label{DynGraphRep}
    \begin{cases}
        \dot{u}^{(n)}_i(t)=\displaystyle\frac{1}{n}\sum_{j=1}^n   A_{ij}^{(n)}(u^{(n)}_j(t)-u^{(n)}_i(t)), \\ 
        u^{(n)}_i(0)=g_i,
    \end{cases}
\end{equation}
where $n$ is the number of nodes and $g_i$ the initial opinion of node $i$. 

The second is the {\em opposing} model, also known as Altafini model, in which the evolution in time of the opinion is instead described by the equation
\begin{equation}\label{DynGraphOpp}
    \begin{cases}
        \dot{u}^{(n)}_i(t)= \displaystyle\frac{1}{n}(\sum_{j=1}^n   A_{ij}^{(n)}u^{(n)}_j(t)-\sum_{j=1}^n|A_{ij}|u^{(n)}_i(t)), \\ 
        u^{(n)}_i(0)=g_i.
    \end{cases}
\end{equation}
Following~\cite{shi2019dynamics}, we observe that the positive interactions are consistent with the classical DeGroot's rule of social interactions,
which postulates that the opinions of trustful social members are attractive to each
other~\cite{degroot1974reaching}. Along a negative link, the opposing rule~\cite{altafini2012consensus} postulates that the interaction will drive a node state to
be attracted by the opposite of its neighbor's state; the repelling rule~\cite{shi2013agreement} indicates
that the two node states will repel each other instead of being attractive.

\subsection{Graphs and graphons}
When studying opinion dynamics on a society scale, the dimension of the graph becomes too large to deal with. To approach this problem, we make use of {\em graphons}. Throughout this paper, we will assume graphons to be symmetric measurable functions $W:[0,1]^2\rightarrow [-1,1]$. We shall instead refer to bounded symmetric measurable functions taking values in $\R$ simply as kernels.

For every kernel $W$, we can define an integral operator $T_W : L^2[0,1]\rightarrow L^2[0,1]$ by:
\[
\left(T_Wf\right)(x):=\int_0^1W(x,y)f(y) \dy .
\]
Notice that $T_{W_1+W_2}=T_{W_1}+T_{W_2}$ for any two kernels $W_1$, $W_2$.
We will equip the space of kernels with the operator norm defined as
\[
\vertiii{T_W}:=\sup_{\|f\|_2=1}\|T_Wf\|_2.
\]

As a tool to compare graphs of different sizes and graphons, it is useful to represent a graph as a piece-wise constant graphon.  To do this, we partition the interval $I=[0,1]$ in $n$ intervals $I_i=(\frac{i-1}{n},\frac{i}{n}]$ for $i=1,\dots,n$, and set $W_n(x,y)=A^{(n)}_{ij},$ for all $(x,y)\in I_i\times I_j$. 
Informally, we can then see a graphon as the limit of a sequence of graphs, as for $n\rightarrow\infty$ the interval of length $\frac{1}{n}$ corresponding to each node has size that goes to 0, so that $[0,1]$ becomes a continuum of nodes.

To make this idea precise, in this paper we will use the distance induced by the operator norm. Hence, we will say that a sequence of graphs (seen as its associated graphons $W_{n}$) converges to a graphon $W$ if $\vertiii{T_{W_n}-T_W}$ goes to zero. In fact, it turns out  that this convergence notion is equivalent to the notion induced by the more popular cut norm; see~\cite[Lemma E.6]{janson2013graphons}. Other (non-equivalent) norms are also used in the literature, such as the $L^2$ norm $\|W_n-W\|_2$ that is used in~\cite{medvedev2014nonlinear}.  

\subsection{Sampled graphs}
The graph sequences of main interest for this paper are sequences of sampled graphs. The latter are widely studied in the literature for graphons with values in $[0,1]$. For values in $[-1,1]$, the natural extension is the following~\cite{borgs2019Lp}.

\begin{defi}[Sampled signed graphs]
    Given a graphon $W$ and a set of points $X_i\in[0,1]$ with $i=1,\ldots,n$, a sampled (signed) graph is a random graph with adjacency matrix $A^{(n)}$ such that
    $$
        A^{(n)}_{ij}=\text{sign}(W(X_i,X_j))\cdot \eta_{ij},
    $$
    where $\eta_{ij}\sim Ber(|W(X_i,X_j)|)$ are independent Bernoulli random variables, for all $i < j$. \label{sampledgraph_negativeweights}
\end{defi}

When we sample from a graphon $W$ with positive values, we have, for instance, that $W(X_i,X_j)=0.5$ corresponds to a 50\% chance of edge $(i,j)$ existing (i.e. $A^{(n)}_{ij}=1$). The definition is then natural extension to negative values, where $W(X_i,X_j)=-0.5$ means there is a 50\% chance of a negative edge existing between nodes $i$ and $j$.

When convenient, with a slight abuse of vocabulary, we will call `graph sampled from the graphon $W$' also the piece-wise constant graphon 
$W_n(x,y)=\sum_{i,j} A^{(n)}_{ij}\mathds{1}_{I_i}(x)\mathds{1}_{I_j}(y)$ associated with the actual graph.

For the set of points $\{X_i\}$, we consider two possible definitions.
\begin{defi}[Latent variables]
For each $n$, consider $n$ points $X_1, \dots, X_n$ in the interval $I$, defined as follows:
\label{latent_var}
    \begin{enumerate}
        \item Deterministic latent variables: $X_i=\frac{i}{n}$, \label{dlv}
        \item Stochastic latent variables: $X_i = U_{(i)}$, where $U_1, \dots, U_n$ are i.i.d.\ uniform in $I$ and $U_{(1)}, \dots, U_{(n)}$ is the corresponding order statistics.\label{slv}
    \end{enumerate}
\end{defi}

\subsection{Opinion dynamics on signed graphons}
In parallel with the graph case, for the {\em repelling} model we can define the dynamics on the graphon as 
\begin{equation} \label{DynGraphonRep}
\begin{cases}
    \displaystyle\frac{\partial u}{\partial t}(x,t)=\displaystyle\int_{I}W(x,y)(u(y,t)-u(x,t))dy  \\
    u(x,0)=g(x),
\end{cases}
\end{equation}
where $u:[0,1]\times\R^+ \to \R$ and $g:[0,1]\to\R$ is the initial opinion distribution, whereas for the {\em opposing} model, we can write 
\begin{equation}\label{DynGraphonOpp}
\begin{cases}\displaystyle
    \frac{\partial u}{\partial t}(x,t)=\int_{I}W(x,y)u(y,t)\dy   - \int_{I}|W(x,y)|u(x,t)dy  \\
    u(x,0)=g(x).
\end{cases}
\end{equation}

Existence and uniqueness of solutions of \eqref{DynGraphonRep} is given in
\cite[Theorem~3.2]{medvedev2014nonlinear}; a slight modification of its proof allows to obtain the same for \eqref{DynGraphonOpp} as well.
\begin{thm}[Existence and uniqueness]  \label{thm:existence-and-unicity}
Assume that $W\in L^{\infty}(I^2)$ and $g\in L^{\infty}(I)$. Then 
there exists a unique solution of \eqref{DynGraphonRep}, $u\in C^1(\mathbb{R}^+;L^{\infty}(I))$.
The same holds for \eqref{DynGraphonOpp}.
\end{thm}
\begin{proof}
The statement concerning \eqref{DynGraphonRep} is \cite[Theorem~3.2]{medvedev2014nonlinear}. The proof for \eqref{DynGraphonOpp} can be obtained with a very similar contraction argument, which we report for completeness.
First, notice that $|W(x,y)| =  W(x,y) + 2 W^-(x,y)$, where
$W^-(x,y) = \max\{-W(x,y),0)\}$ is the negative part of $W$.
With this, we can see that the integro-differential equation in \eqref{DynGraphonOpp} is equivalent to the integral equation $u = Ku$ with $K$ defined as 
\begin{multline*}
  (Ku)(x,t)  = g(x) + \int_0^t \int_I |W(x,y)| (u(y,s)-u(x,s)) \dy \, \mathrm ds\\
+ 2 \int_0^t \int_I W^-(x,y) u(y,s) \dy \, \mathrm ds    .
\end{multline*}
Consider $M_g$   the space of functions $u \in C([0,\tau];L^{\infty}(I))$ such that $u(x,0) = g(x)$, with the norm 
$\|v \|_{M_g} = \max_{t \in [0, \tau]}  \esssup_{x \in I} |v(x,t)| $.
Clearly $K: M_g \to M_g$. The goal is to show that $K$ is a contraction, for some sufficiently small $\tau$.
From the definition of $K$, 
\begin{multline*}
 (Ku - Kv)(x,t) =  \int_0^t \int_I |W(x,y)| (u(y,s)-v(y,s))  \dy \, \mathrm ds\\
 \qquad\qquad\qquad+ \int_0^t \int_I |W(x,y)| (-u(x,s)+v(x,s)) \dy \, \mathrm ds\\
+ 2 \int_0^t \int_I W^-(x,y) (u(y,s)-v(y,s)) \dy \, \mathrm ds    .
\end{multline*}
from which is is easy to see that 
\begin{multline*}
 |(Ku - Kv)(x,t) | \leq
  \| W \|_{\infty} \Big [ 3 \int_0^t \int_I | u(y,s)-v(y,s)|  \dy \, \mathrm ds \\
 +   \int_0^t |u(x,s)-v(x,s)| \, \mathrm ds \Big ] .
\end{multline*} 
This bound is the same as in the proof of \cite[Theorem~3.2]{medvedev2014nonlinear}, except for the factor 3 in the first term in the right-hand side. The conclusion then follows in a similar way: 
\begin{multline*} 3 \int_0^t \int_I | u(y,s)-v(y,s)|  \dy \, \mathrm ds 
 +   \int_0^t |u(x,s)-v(x,s)| \, \mathrm ds\\
 \leq 4\,  t \max_{s\in [0,t]} \| u(\cdot,s)-v(\cdot,s)\|_{L^\infty (I)} \, ,
 \end{multline*} 
 which does not depend on $x$,
 and hence
\begin{multline*}  \|  Ku - Kv \|_{M_g}  \leq 
4 \| W \|_{\infty} \max_{t\in[0,\tau]} t \max_{s\in [0,t]} \| u(\cdot,s)-v(\cdot,s)\|_{L^\infty (I)} \\
\leq  4 \| W \|_{\infty} \tau \max_{t\in[0,\tau]} \| u(\cdot,t)-v(\cdot,t)\|_{L^\infty (I)} 
=  4 \| W \|_{\infty} \tau \| u - v \|_{M_g} .
\end{multline*} 
Choosing $\tau = 1 / (8 \| W \|_{\infty})$, this ensures that $K$ is a contraction.
By Banach contraction theorem, this ensures existence and uniqueness of solution $u \in M_g$. The argument can then be iterated on further intervals $[k \tau, (k+1) \tau]$ to extend the result to the real line.
Furthermore, since the integrand (in time) in the definition of $K$ is continuous, $u$ is continuously differentiable. Thus, we have a classical solution. 
\end{proof}

\section{Convergence of solutions}\label{sect:convergence-solutions}
The objective of this section is to compare the solutions of \eqref{DynGraphRep} to the solutions of \eqref{DynGraphonRep} (and analogously for solutions of \eqref{DynGraphOpp} and of \eqref{DynGraphonOpp}) for large values of $n$. Since these solutions belong to different spaces, it is necessary to explain how this comparison will be made. The following lemma is instrumental to allow for a fair comparison.

\begin{lem}[Graph dynamics as graphon dynamics] \label{lem:equivalenza-discreto}
Define
$W_n(x,y)=\sum_{i,j} A^{(n)}_{ij}\mathds{1}_{I_i}(x)\mathds{1}_{I_j}(y)$,
$g_n(x) = \sum_{i=1}^n  g_i \mathds{1}_{I_i}(x)$,
and
$u_n(x,t)=\sum_{i=1}^n u^{(n)}_i(t)\mathds{1}_{I_i}(x)$,
where $I_i=(\frac{i-1}{n},\frac{i}{n}]$, $\mathds{1}_{I_i}(x)=1$ if $x\in I_i$ and 0 otherwise. 

If  $u^{(n)}(t)$ is solution of \eqref{DynGraphRep}, then $u_n(x,t)$ is  solution of 
\begin{equation}  \label{DynGraphRep-integrodiff}
\begin{cases}
    \frac{\partial u_n}{\partial t}(x,t)=\int_{I}W_n(x,y)(u_n(y,t)-u_n(x,t))\dy  \\
    u_n(x,0)=g_n(x). 
\end{cases}
\end{equation}

Analogously, if $u^{(n)}(t)$ is solution of \eqref{DynGraphOpp}, then $u_n(x,t)$ is solution of 
\begin{equation}  \label{DynGraphOpp-integrodiff}
\begin{cases}\displaystyle
    \frac{\partial u_n}{\partial t}(x,t)=\int_{I}W_n(x,y)u_n(y,t)\dy   - \int_{I}|W_n(x,y)|u_n(x,t)dy  \\
    u_n(x,0)=g_n(x).
\end{cases}
\end{equation}
\end{lem}
\begin{proof}
We detail the proof for the repelling model.
From the definition of $u_n(x,t)$, we have
$$
\frac{\partial u_n(x,t)}{\partial t}=\sum_{i=1}^n 
\frac{\mathrm du^{(n)}_i(t)}{\mathrm dt} \, \mathds{1}_{I_i}(x),
$$
which means that, for any $i$, for all $x\in I_i$,
\begin{align*}
    \frac{\partial u_n(x,t)}{\partial t}=\frac{\mathrm du^{(n)}_i(t)}{\mathrm dt}&=
    \sum_{j=1}^n \frac{1}{n} \, A_{ij}(u^{(n)}_j(t)-u^{(n)}_i(t))\\
    &= \sum_{j=1}^n \int_{I_j}W_n(x,y)(u_n(y,t)-u_n(x,t))\dy \\
    &= \int_{I} W_n(x,y)(u_n(y,t)-u_n(x,t))\dy. 
\end{align*}
The proof for the opposing model is analogous. 
\end{proof}
Notice that \eqref{DynGraphRep-integrodiff} and \eqref{DynGraphOpp-integrodiff} are particular cases of \eqref{DynGraphonRep} and \eqref{DynGraphonOpp}, respectively, where the graphon $W$ and the initial condition $g$ are piece-wise constant. Hence, Theorem~\ref{thm:existence-and-unicity} on existence and uniqueness of classical solutions applies to \eqref{DynGraphRep-integrodiff} and \eqref{DynGraphOpp-integrodiff} as well.
This implies that \eqref{DynGraphRep} and \eqref{DynGraphOpp} respectively have the equivalent formulations \eqref{DynGraphRep-integrodiff} and \eqref{DynGraphOpp-integrodiff} which use graphons, so that we can now directly compare $u(x,t)$ and $u_n(x,t)$. 
We are now ready to state the main result of this section.
We will use the notation  $ \| (u-u_n)(\cdot,t)\|_2$ for the $L^2$ norm 
of $u-u_n$ in the $x$ variable only.

\begin{thm}[Convergence error estimate]\label{thm:bound}
Consider $W: I^2 \to [-1,1]$ symmetric and measurable, 
$g \in L^\infty(I)$, 
$W_n$ and $g_n$ as in Lemma~\ref{lem:equivalenza-discreto}.

If $u_n$ and $u$ are solutions of  \eqref{DynGraphRep-integrodiff} and \eqref{DynGraphonRep}, respectively,  then for all $t \in [0,T]$
\begin{equation*} 
    \| (u-u_n)(\cdot,t)\|_2^2
    \leq \left(  \|g-g_n\|_2+ C_u \vertiii{T_{W}-T_{W_n}} \right) \exp(2 T),
\end{equation*}
where $C_u = \displaystyle \esssup_{t\in[0,T], x\in I} |u(x,t)|$ (which is finite by Theorem~\ref{thm:existence-and-unicity}).

If $u_n$ and $u$ are solutions of  \eqref{DynGraphOpp-integrodiff} and \eqref{DynGraphonOpp}, respectively,  then for all $t \in [0,T]$
\begin{multline*}    \| (u-u_n)(\cdot,t)\|_2^2 \leq \\
    \left[ \|g-g_n\|_2 + 
         C_u \! \left( \vertiii{T_{W^+} -T_{W_n^+}} \!+\! \vertiii{T_{W^-} -T_{W_n^-}}  \right) \! \right] \exp(4 T),
\end{multline*}
where $W = W^+ - W^-$ denotes the decomposition of $W$ in its positive part $W^+ = \max(W,0)$ and its negative part $W^- = \max(-W, 0)$.
\end{thm}
\begin{proof} 
    We will use the short-hand notations 
    $\xi_n(x,t)=u_n(x,t)-u(x,t)$
    and
    $\tilde W_n(x,y) = W_n(x,y) - W(x,y)$. 
From \eqref{DynGraphonRep} and \eqref{DynGraphRep-integrodiff}, we have 
    \begin{multline*}
        \frac{\partial \xi_n(x,t)}{\partial t}
        = \int_I W_n(x,y)(u_n(y,t)-u_n(x,t)) \dy\\
        - \int_I W(x,y) (u(y,t)-u(x,t)) \dy.
    \end{multline*}
We add  the following zero term to the right-hand side:
$-\int_I W_n(x,y) (u(y,t)-u(x,t)) \dy 
+ \int_I W_n(x,y) (u(y,t)-u(x,t)) \dy $,
to obtain
    \begin{multline*}
        \frac{\partial \xi_n(x,t)}{\partial t}
        = \int_I W_n(x,y)(\xi_n(y,t)-\xi_n(x,t)) \dy\\
        - \int_I \tilde W_n(x,y) (u(y,t)-u(x,t)) \dy.
    \end{multline*}
Then we multiply  
both sides by $\xi_n(x,t)$ and integrate over $I$ 
\begin{multline}    \label{eq:proof-convergence-int}
\int_I \frac{\partial \xi_n(x,t)}{\partial t} \,\xi_n(x,t) \dx =\\
\int_{I^2} W_n(x,y)(\xi_n(y,t)-\xi_n(x,t))\xi_n(x,t) \dx \dy\\
    +\int_{I^2} \tilde W_n(x,y)(u(y,t)-u(x,t))\xi_n(x,t) \dx\dy.
\end{multline}

For the left-hand side, notice that
\begin{equation} \label{eq:proof-convergence-int-left}
\int_I \frac{\partial \xi_n(x,t)}{\partial t} \xi_n(x,t) \dx 
= \frac{1}{2} \int_I \frac{\partial (\xi_n(x,t))^2}{\partial t}  \dx
= \frac{1}{2} \frac{\mathrm d}{\mathrm d t} \| \xi_n(\cdot,t)\|_2^2.
\end{equation}

By using $\|W\|_{\infty}\le 1$, the triangle inequality and the Cauchy-Schwartz inequality, we estimate the first term of the right-hand side of~\eqref{eq:proof-convergence-int} as follows:
\begin{multline} \label{eq:proof-convergence-int-right1} \Big| \int_{I^2} W_n(x,y)(\xi_n(y,t)-\xi_n(x,t))\xi_n(x,t) \dx \dy \Big|    
\leq \\
%
\leq \int_{I^2}|(\xi_n(y,t)-\xi_n(x,t))\xi_n(x,t)|\dx\dy\leq 2 \, \|\xi_n(\cdot, t)||^2_2.
\end{multline}

For the second term of the right-hand side of \eqref{eq:proof-convergence-int}, 
\begin{multline}  \label{eq:proof-convergence-right2-preliminary}
  \Big| \int_{I^2} \tilde W_n(x,y)(u(y,t)-u(x,t))\xi_n(x,t) \dx\dy \Big| \\
\leq 
  \Big| \int_{I^2} \tilde W_n(x,y) u(y,t) \xi_n(x,t) \dx\dy \Big| \\
  + \Big| \int_{I^2} \tilde W_n(x,y) u(x,t) \xi_n(x,t) \dx\dy \Big| .
\end{multline}
Recalling the definition of the integral operator $T_W$ associated with a graphon $W$, we have
\begin{multline*}
     \Big| \int_{I^2} \tilde W_n(x,y) u(y,t) \xi_n(x,t) \dx\dy \Big| \\=
\Big| \int_{I} \left( T_{\tilde W_n} \xi_n(\cdot,t) \right)(y) u(y,t)  \dy \Big| 
\leq \| (T_{\tilde W_n} \xi_n)(\cdot,t) \, u(\cdot,t)  \|_1 .
\end{multline*}
Then, by H\"older inequality,
\begin{multline*}
\| (T_{\tilde W_n} \xi_n)(\cdot,t) \, u(\cdot,t)  \|_1 
\leq \| (T_{\tilde W_n} \xi_n)(\cdot,t) \|_2  \, \| u(\cdot,t)  \|_2  \\
\leq \vertiii{T_{\tilde W_n}} \, \|\xi_n(\cdot,t) \|_2 \, \| u(\cdot, t) \|_2
\leq \vertiii{T_{\tilde W_n}} \, \|\xi_n(\cdot,t) \|_2 \, \| u(\cdot, t) \|_\infty.
\end{multline*}
Similarly, for the second term in \eqref{eq:proof-convergence-right2-preliminary},
\begin{multline*}
     \Big| \int_{I^2} \tilde W_n(x,y) u(x,t) \xi_n(x,t) \dx\dy \Big| \\=
\Big| \int_{I} \left( T_{\tilde W_n} (\xi_n(\cdot,t) u(\cdot,t)) \right)(y)  \dy \Big| 
\leq \|  T_{\tilde W_n} (\xi_n(\cdot,t) u(\cdot,t))  \|_1 ,
\end{multline*}
and
\begin{multline*}
\| T_{\tilde W_n} (\xi_n(\cdot,t) u(\cdot,t)) \|_1 
\leq \| T_{\tilde W_n} (\xi_n(\cdot,t) u(\cdot,t)) \|_2  \\
\leq \vertiii{T_{\tilde W_n}} \, \|\xi_n(\cdot,t) u(\cdot,t)\|_2 
\leq \vertiii{T_{\tilde W_n}} \, \|\xi_n(\cdot,t) \|_2 \, \| u(\cdot, t) \|_\infty.
\end{multline*}
We have obtained the following bound for the second term in \eqref{eq:proof-convergence-int}:
\begin{multline}
    \label{eq:proof-convergence-right2}
  \Big| \int_{I^2} \tilde W_n(x,y)(u(y,t)-u(x,t))\xi_n(x,t) \dx\dy \Big| \\
  \leq
2 \vertiii{T_{\tilde W_n}} \, \|\xi_n(\cdot,t) \|_2 \, \| u(\cdot, t) \|_\infty
  \leq 2 \, C_u \,  \vertiii{T_{\tilde W_n}} \, \|\xi_n(\cdot,t) \|_2 ,
\end{multline}
with $C_u = \esssup_{t \in [0,T], x\in I} |u(x,t) | $,
which is finite, since $u$ is continuous in time and essentially bounded in space, by Theorem~\ref{thm:existence-and-unicity}.

Combining \eqref{eq:proof-convergence-int}, 
\eqref{eq:proof-convergence-int-left}, \eqref{eq:proof-convergence-int-right1} and \eqref{eq:proof-convergence-right2}, we have
\begin{equation}  \label{eq:proof-convergence-final}
   \frac{1}{2} \frac{\mathrm d}{\mathrm d t} \| \xi_n(\cdot,t)\|_2^2
   \leq 2 \,  \| \xi_n(\cdot,t)\|_2^2
   + 2 \, C_u \, \vertiii{T_{\tilde W_n}} \, \|\xi_n(\cdot,t) \|_2 .
\end{equation}
This is  similar to the estimate obtained in \cite{medvedev2014nonlinear}, but with different constants and $\vertiii{T_{\tilde W_n}}$ instead of $\|\tilde W_n\|_2$.
Then, the proof can then concluded exactly as in \cite{medvedev2014nonlinear}, as follows.
To `divide by $\| \xi_n(\cdot, t) \|_2$' avoiding the risk to divide by zero, one
defines an auxiliary function 
$\phi_{\eps}(t)=\sqrt{\| \xi_n(\cdot, t) \|_2^2+\eps}$, for any $\eps>0$.
From \eqref{eq:proof-convergence-final},
\[
     \frac{1}{2} \frac{\mathrm d}{\mathrm dt}\phi_{\eps}(t)^2 \leq 
     2 \phi_{\eps}(t)^2 + 2 C_u \, \vertiii{T_{\tilde W_n}} \phi_{\eps}(t)
\]
and, by chain rule and dividing by $\phi_{\eps}(t)$, 
\[
\frac{\mathrm d}{\mathrm dt}\phi_{\eps}(t)
\leq  2 \phi_{\eps}(t) + 2 C_u\, \vertiii{T_{\tilde W_n}} .
\]
By Grönwall's lemma, for all $t \in [0,T]$
\[  \phi_{\eps}(t)
\leq 
\left(\phi_{\eps}(0)+ C_u \vertiii{T_{\tilde W_n}} \right)\exp(2 \, T).
\]
Since $\epsilon$ is arbitrary, this gives the conclusion.



Now we consider the opposing case. 
From \eqref{DynGraphonOpp} and \eqref{DynGraphOpp-integrodiff}, 
recalling that $|W_n(x,y)| = W_n(x,y) + 2 W_n^-(x,y)$,
we have 
    \begin{multline*}
        \frac{\partial \xi_n(x,t)}{\partial t}
        = \int_I |W_n(x,y)| \, (u_n(y,t)-u_n(x,t)) \dy\\
        - \int_I |W(x,y)| \, (u(y,t)-u(x,t)) \dy \\
        - 2 \int_I W_n^-(x,y) u_n(y,t) \dy
        + 2 \int_I W^-(x,y) u(y,t) \dy.
    \end{multline*}
We add  the following zero term to the right-hand side:
$-\int_I W_n(x,y) (u(y,t)-u(x,t)) \dy 
+ \int_I W_n(x,y) (u(y,t)-u(x,t)) \dy 
- 2 \int_I W_n^-(x,y) u(y,t) \dy 
+ 2 \int_I W_n^-(x,y) u(y,t) \dy$.
Then we multiply both sides by $\xi_n(x,t)$ and integrate in $I$ in the $x$ variable. 
We obtain:
\begin{align}
\label{eq:proof-convergence-int_Opp}
\nonumber  \int_I &\frac{\partial \xi_n(x,t)}{\partial t} \,\xi_n(x,t) \dx\\
\nonumber=&\int_{I^2} |W_n(x,y)| \, (\xi_n(y,t)-\xi_n(x,t))\xi_n(x,t) \dx \dy\\
\nonumber    &+\!\!\int_{I^2} ( |W_n(x,y)| - |W(x,y)| )(u(y,t)-u(x,t))\xi_n(x,t) \dx\dy\\
\nonumber    &+ 2 \int_{I^2} W_n^-(x,y) \xi_n(y,t) \xi_n(x,t) \dx\dy\\
    &+ 2 \int_{I^2} ( W_n^-(x,y) - W^-(x,y) )  u(y,t)  \xi_n(x,t) \dx\dy.
\end{align}

Compare \eqref{eq:proof-convergence-int_Opp} with  \eqref{eq:proof-convergence-int}.
The left-hand side is the same 
and the first two terms of the right-hand side are almost the same, with the only difference that $W_n(x,y)$ and $W(x,y)$ 
are now replaced by $|W_n(x,y)|$ and $|W(x,y)|$. 
Hence, \eqref{eq:proof-convergence-int-left}, \eqref{eq:proof-convergence-int-right1}
and \eqref{eq:proof-convergence-right2} still apply, except that 
$\vertiii{T_{\tilde W_n}}$ is replaced by 
$\vertiii{T_{|W_n| - |W|}}$.
In the following, we take care of the remaining terms, which were not present in \eqref{eq:proof-convergence-int}.

Since $|W^-(x,y)| \leq 1$,
\begin{multline*}
   \Big| 2 \int_{I^2} W_n^-(x,y) \xi_n(y,t) \xi_n(x,t) \dx\dy \Big| \\
    \leq 2 \int_{I^2} | \xi_n(y,t) \xi_n(x,t) | \dx\dy
    = 2 \| \xi_n(\cdot, t) \|_2^2.
\end{multline*}
For the last term,
\begin{align*}
   &\Big| 2 \int_{I^2}  ( W_n^-(x,y) - W^-(x,y) )  u(y,t)  \xi_n(x,t) \dx\dy \Big| \\
   &\leq 2 \int_I |(T_{W_n^- - W^-} \xi_n)(\cdot,t))(y)| \, |u(y,t)| \dy \\
   &= 2 \|(T_{W_n^- - W^-} \xi_n)(\cdot,t) \, u(\cdot,t)\|_1\\
   &\leq 2 \|(T_{W_n^- - W^-} \xi_n)(\cdot,t) \|_2  \, \|  u(\cdot,t)\|_{\infty} \\
   &\leq 2 \vertiii{T_{W_n^- - W^-}} \, \|\xi_n(\cdot,t) \|_2  \, C_u.
\end{align*}
Collecting the upper bounds for all four terms in the right-hand side of \eqref{eq:proof-convergence-int_Opp}, we get
\begin{multline*}
\frac{1}{2} \frac{\mathrm d}{\mathrm d t} \| \xi_n(\cdot,t)\|_2^2 
   \,\leq\,  4 \,  \| \xi_n(\cdot,t)\|_2^2 \\
   + 2 \, C_u \, \left( \vertiii{T_{|W_n| - |W|}} + \vertiii{T_{W_n^- - W^-}} \right)\, \|\xi_n(\cdot,t) \|_2.
\end{multline*}
Using $|W| = W^+  +  W^-$, 
$T_{W_1+W_2} = T_{W_1} + T_{W_2}$ 
and triangle inequality, we have
$\vertiii{T_{|W_n|-|W|}} \leq \vertiii{T_{W_n^+} -T_{W^+}} + \vertiii{T_{W_n^-} -T_{W^-}} $. Hence,
\begin{multline*}
\frac{1}{2} \frac{\mathrm d}{\mathrm d t} \| \xi_n(\cdot,t)\|_2^2 
   \,\leq\,  4 \,  \| \xi_n(\cdot,t)\|_2^2 \\
   + 4 \, C_u \, \left( \vertiii{T_{W_n^+} -T_{W^+}} + \vertiii{T_{W_n^-} -T_{W^-}} \right)\, \|\xi_n(\cdot,t) \|_2,
\end{multline*}
which is the analogous of \eqref{eq:proof-convergence-final}. 
The conclusion then follows by the same technique involving $\phi_\eps$ and Grönwall's lemma. 
\end{proof}

Theorem~\ref{thm:bound} implies that, so long as 
the initial conditions converge and the finite graphs converge to the graphon, the dynamics on the finite graphs converge to the dynamics on the graphon, in the sense that solutions converge on bounded intervals. 

\begin{rem}[Choice of the convergence norm]
    The proof of convergence as $n$ goes to infinity is based on adapting arguments from \cite{medvedev2014nonlinear}.
Other than some extra terms involved in the opposing Laplacian, the main differences with \cite{medvedev2014nonlinear} are that we restrict ourselves to linear dynamics and that we obtain a stronger upper bound, which involves $\vertiii{T_{W-W_n}}$, whereas \cite{medvedev2014nonlinear} considers $\| W-W_n \|_2$. Having the operator norm is essential to apply the result to sampled graphs, for which $\vertiii{T_{W-W_n}}$ vanishes for large $n$ under mild assumptions, while this is not the case for $\| W-W_n \|_2$.
\end{rem}

\section{Convergence of solutions for sampled graphs}
\label{sect:convergence-sampled}

In this section, we apply Theorem~\ref{thm:bound} to prove that solutions to opinion dynamics on large graphs that are sampled from graphons are well approximated by solutions to opinion dynamics on graphons, over finite time horizons. 

In order to have convergence of the solutions of the two different dynamics, we need to have compatible initial conditions. 
Given the initial condition $g(x)$ for the dynamics on the graphon, we choose to define the initial condition of the dynamics on the sampled graph $g_n(x)$ as a piece-wise constant function by assigning the value $g(X_i)$ to  the entire interval $I_i$. We now show that with this definition, the term $\|g_n-g\|_2$ present in Theorem \ref{thm:bound} goes to 0 as $n$ increases.
\begin{prop}[Convergence of sampled initial conditions]\label{prop:conv_in_cond}
    Given $g: I \to \R$ bounded and almost everywhere continuous,
    and given $X_1, \dots, X_n \in I$, define the piece-wise constant function
    $g_n(x)=g(X_i)$ for all $x\in I_i$.
If $X_1, \dots, X_n$ are deterministic latent variables, 
 then $\lim_{n \to \infty} \|g_n-g\|_2 =0$.
If $X_1, \dots, X_n$ are stochastic latent variables, 
then
$\lim_{n \to \infty} \|g_n-g\|_2 =0$ almost surely.
 \end{prop}

 \begin{proof}
 Preliminarily, observe that since $g$ is almost everywhere continuous, for almost all $x$ we have
 $\lim_{h \to 0} |g(x+h)-g(x)| = 0$.
In the case of deterministic latent variables, for all $i$, for all $x \in I_i$, $g_n(x) = g(X_i)$, with $|X_i - x| \le \frac{1}{n}$, which goes to zero when $n \to \infty$. 
Hence, for almost all $x$ we have  $\lim_{n \to \infty} |g_n(x)-g(x)| = 0$.
From this pointwise (almost everywhere) convergence we can obtain convergence in $L^2$ norm, using the Lebesgue dominated convergence theorem. 
Indeed, since $g $ is bounded, there exists $m_g > 0$ s.t.\ $|g(x)| \le m_g$ for all $x$ and hence also $|g_n(x)| \le m_g$ for all $x$, so that $|g_n(x) - g(x)| \le 2 m_g$.
We have that  $(g_n - g)^2$ converges pointwise a.e.\ to $0$ and is bounded by a constant, so that by  Lebesgue dominated convergence theorem
$ \int_I (g_n(x) - g(x))^2 \dx $ converges to zero.

The proof in the case of stochastic latent variables is similar, the only difference being in the step that shows that for $x \in I_i$, $|X_i - x|$ goes to zero when $n \to \infty$, which now can be obtained a.s., as follows.
We use \cite[Prop.~3]{avella2018centrality} (taking $\delta = 1/n^2)$: for all sufficiently large $n$, with probability at least $1 - \frac{1}{n^2}$,
\[ \left | X_i - \frac{i}{n+1}\right |  \le \sqrt{\frac{24 \log n}{n+1}} \,, \; \forall i = 1, \dots , n.\]
By Borel-Cantelli Lemma, this further implies that almost surely the same bound holds 
eventually.
Now simply notice that  for all $i$, for all $x \in I_i$
\[ | X_i - x | \le 
| X_i - \tfrac{i}{n+1}  | + |  \tfrac{i}{n+1} - \tfrac{i}{n} | 
    + |\tfrac{i}{n}  - x| \]
with $|\tfrac{i}{n}  - x| \le \frac{1}{n}$. Hence, almost surely, for $n \to \infty$
\[ | X_i - x | = O\left(  \sqrt{\frac{24 \log n}{n+1}} + \frac{2}{n} \right)   \,,\]
for all $i$ and all $x \in I_i$.
\end{proof}


To show that the second terms of the right-hand sides of the bounds in Theorem \ref{thm:bound} go to 0, we 
have
the following proposition, where for simplicity we assume that the graphon is piece-wise Lipschitz~\cite[Assumption~1]{avella2018centrality}. 
\begin{prop}[Convergence of sampled graphs]
\label{convergence_sampled_to_graphon}
    If $W:[0,1]\to[-1,1]$ is a piece-wise Lipschitz graphon and $W_n$ is a graph sampled from $W$ with either deterministic or stochastic latent variables $X_i$, then almost surely $\vertiii{T_{W_n}-T_W}\to0$,
    $\vertiii{T_{W^+_n}-T_{W^+}}\to0$ and $\vertiii{T_{W^-_n}-T_{W^-}}\to~0$ for $n\to\infty$.
\end{prop}

\begin{proof}
    The result is well-known \cite[Thm~1]{avella2018centrality} when $W$ has values in $[0,1]$. To extend it to our case, 
    we use the decomposition in positive and negative parts, $W=W^+-W^-= \max(W,0)-\max(-W, 0)$.
    Notice that the construction of the sampled graph $W_n$ in Definition~\ref{sampledgraph_negativeweights} is equivalent of building a sampled graph $W_n^+$ from the graphon $W^+$, a sampled graph $W_n^-$ from the graphon $W^-$, and defining $W_n=W_n^+-W_n^-$.
In fact, $W_n^+$ can be seen either as a graph sampled from the positive part of $W$ or as the positive part of a signed graph sampled from $W$ (and the same goes for $W_n^-$).
Notice that $W^+$ and $W^-$ are standard graphons with values in [0,1], and hence, by \cite[Thm~1]{avella2018centrality} and $\vertiii{T_{W_n^+}-T_{W^+}}\to 0$, $\vertiii{T_{W_n^-}-T_{W^-}}\to 0$ almost surely for $n \to \infty$.
The almost sure convergence of $ \vertiii{T_{W_n}-T_{W}} $ to zero then follows by triangular inequality:
\begin{equation*}\label{eq:conv-sampled-graphs}
\begin{split}
    \vertiii{T_{W_n}-T_{W}}   
    &=\vertiii{T_{W_n^+}-T_{W_n^-}-T_{W^+}+T_{W^-}}\\
    &\leq\vertiii{T_{W_n^+}-T_{W^+}}+\vertiii{T_{W_n^-}-T_{W^-}}.
\end{split}
\end{equation*}

\end{proof}

Propositions~\ref{prop:conv_in_cond} and~\ref{convergence_sampled_to_graphon} show that, under the given assumptions, all the norms in the right-hand sides of the bounds in Theorem \ref{thm:bound} go to 0, which proves the following theorem.
\begin{thm}[Convergence of solutions on sampled graphs]\label{thm:sampled-convergence}
Let $W:[0,1] \to [-1,1]$ be piece-wise Lipschitz and let $u$ be a solution to \eqref{DynGraphonRep} (respectively, to \eqref{DynGraphonOpp}). For each $n\in\mathbb{N}$, 
consider deterministic or stochastic latent variables as in Definition~\ref{latent_var} and
let $W_n$ be a graph of size $n$, sampled from $W$ as per Definition~\ref{sampledgraph_negativeweights}, and $u_n$ be the solution to \eqref{DynGraphRep} (respectively, to \eqref{DynGraphOpp}). Assume that the initial conditions satisfy the assumptions of  
Proposition~\ref{prop:conv_in_cond}.
Then,
for any fixed $T >0$, for $n \to \infty$,
\[ \max_{t \in [0,T]}  \|u_n(\cdot, t)-u(\cdot,t)\|_2\to 0 \,  \text{ a.s.} \]
\end{thm}

\section{Simulations}\label{sect:simulations}

\begin{figure}  
    \centering
    \includegraphics[width=0.8\linewidth]{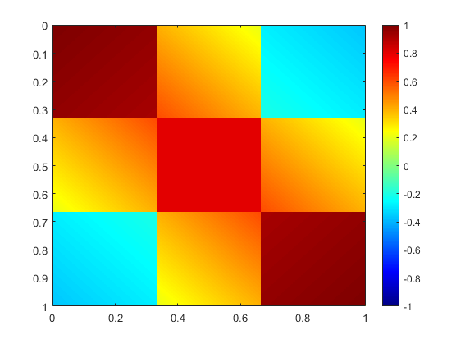}
    \caption{The graphon used in the simulations.} 
    \label{W_plot}
\end{figure}

\begin{figure*}
    \includegraphics[width=0.65\columnwidth]{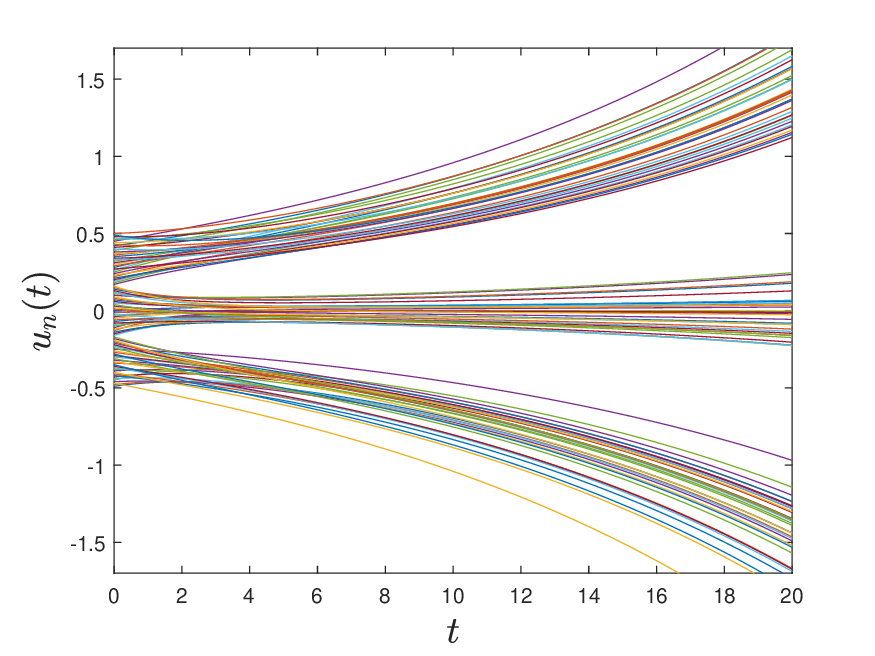}
    \includegraphics[width=0.65\columnwidth]{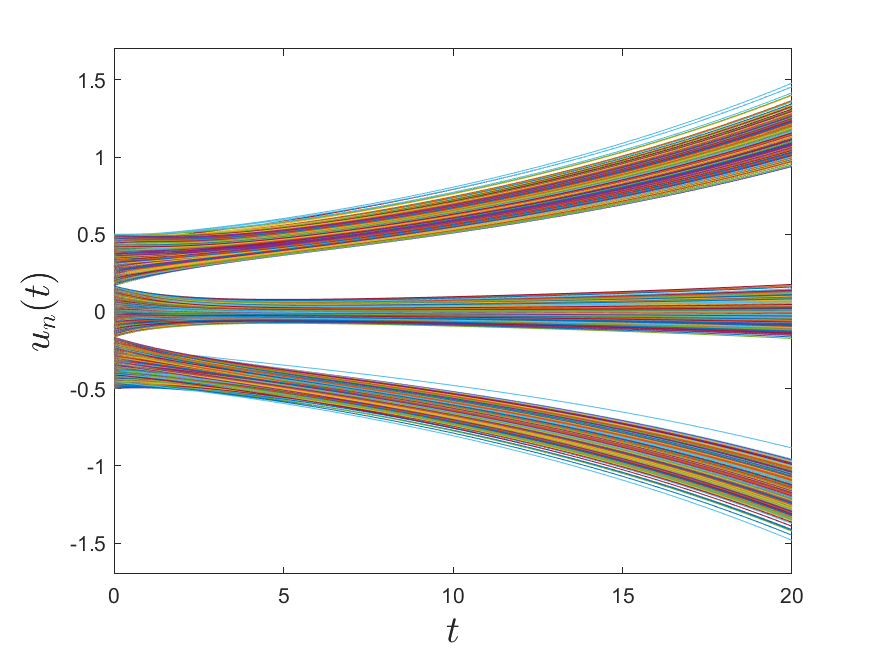}
    \includegraphics[width=0.65\columnwidth]{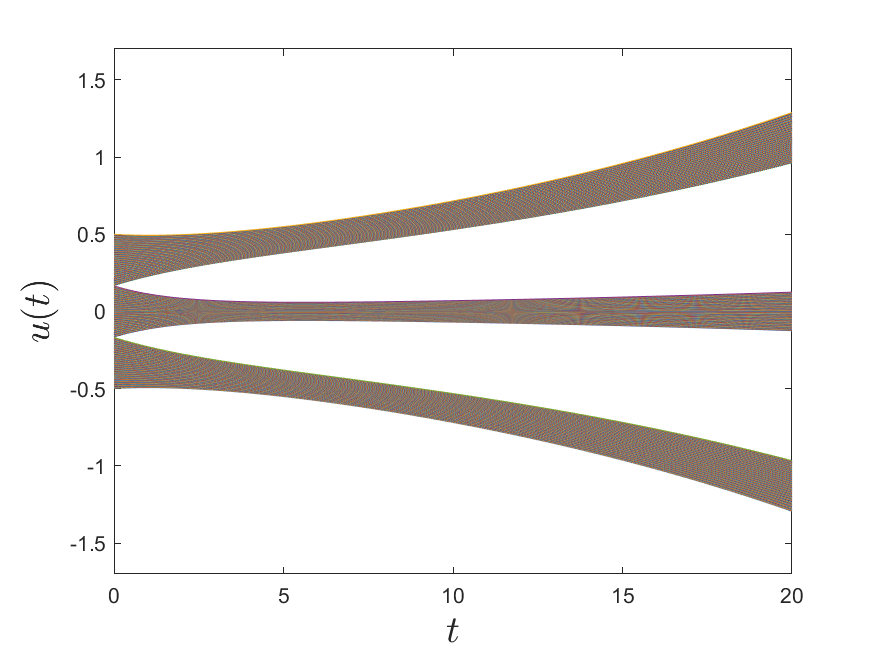}
    \caption{Comparison between solutions on graphs with $n=100$ (left) and $n=1000$ (center) and on the graphon (right) for the repelling dynamics, with initial condition $g(x)=x-\frac12$.} 
    \label{fig:diverging-solutions}
\end{figure*}

\begin{figure*}
    \includegraphics[width=0.65\columnwidth]{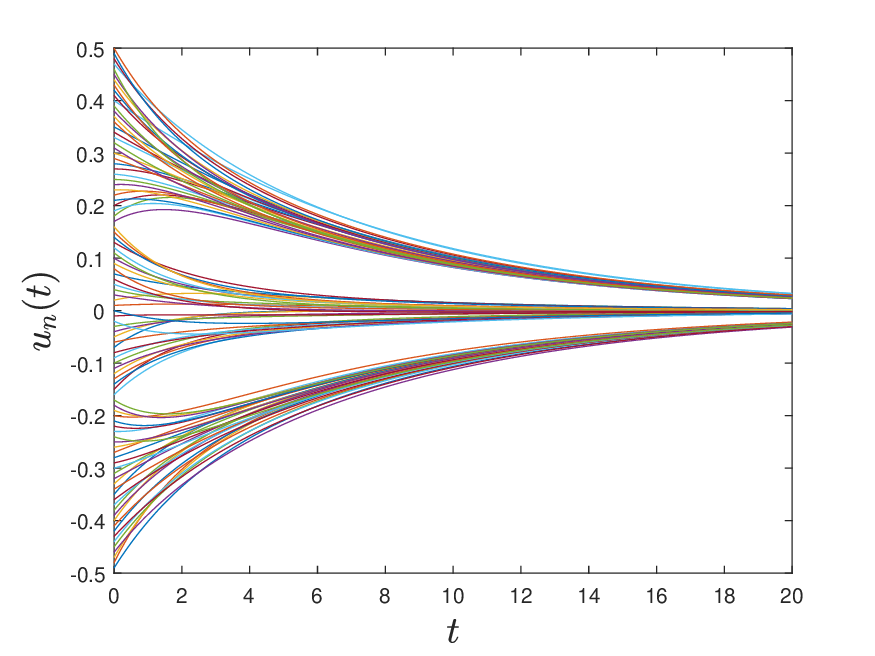}
    \includegraphics[width=0.65\columnwidth]{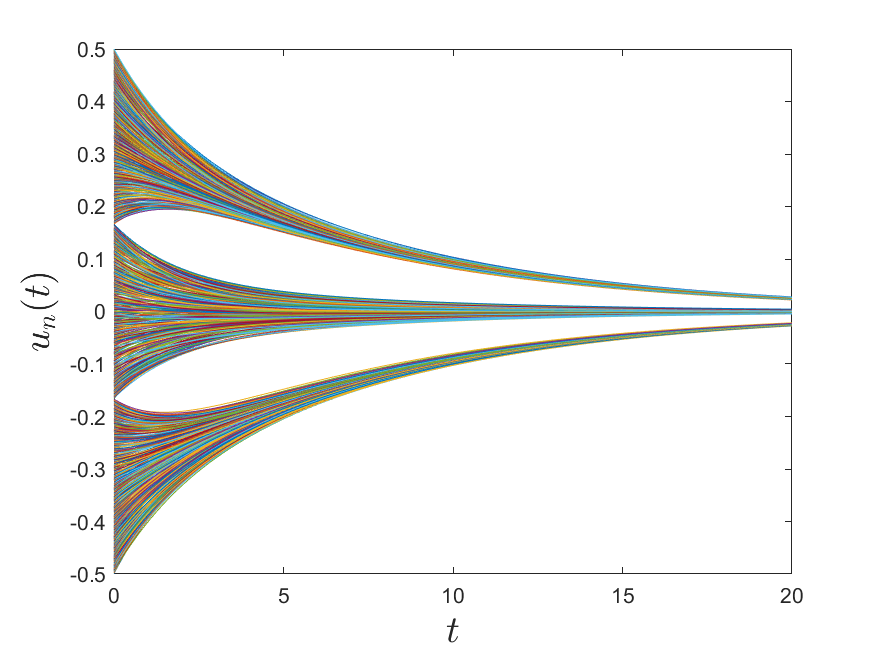}
    \includegraphics[width=0.65\columnwidth]{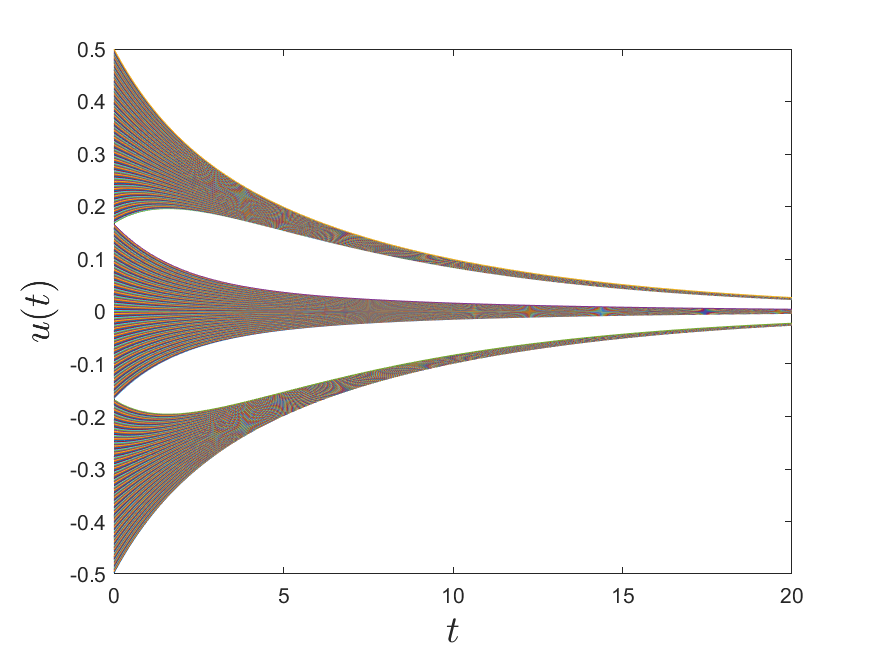}
    \caption{Comparison between solutions on graphs with $n=100$ (left) and $n=1000$ (center) and on the graphon (right) for the opposing dynamics, with initial condition $g(x)=x-\frac12$.} 
    \label{fig:opposing}
\end{figure*}
\begin{figure*}
    \centering
    \includegraphics[width=0.45\linewidth]{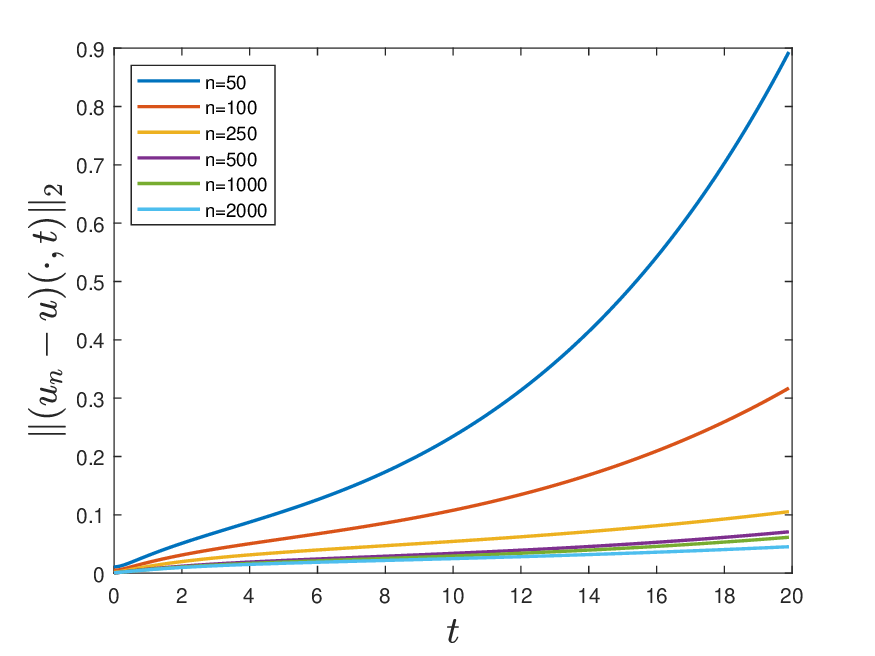}
    \includegraphics[width=0.45\linewidth]{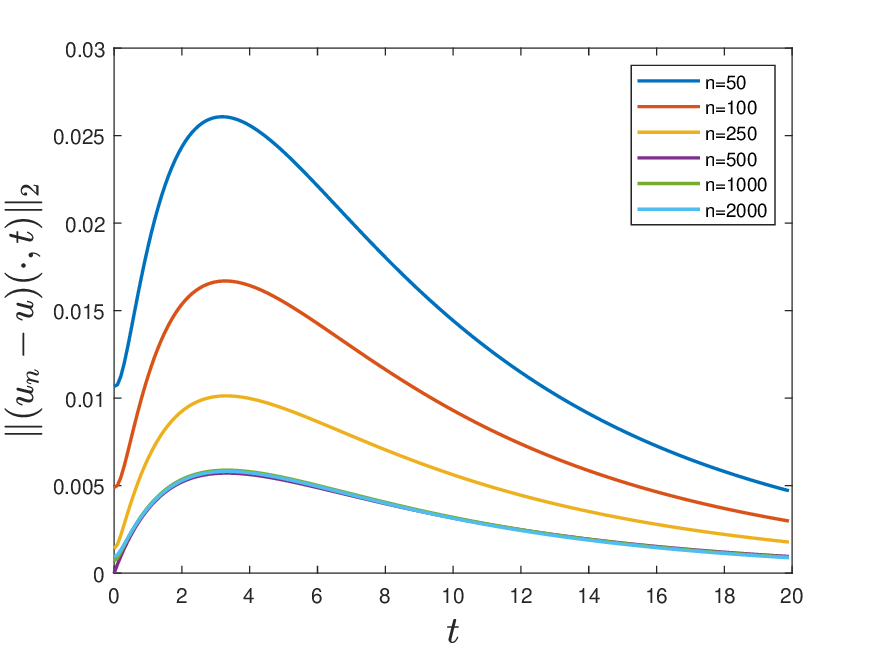}
    \caption{Evolution in time of $\|(u_n-u)(\cdot,t)\|_2$ for the repelling (left) and the opposing dynamics (right).}
    \label{fig:err_L2}
\end{figure*}
In this section we investigate through simulations an example of graphon dynamics. We take inspiration from \cite{aletti2022opinion}, where the authors assume the network to be a stochastic block model describing the interactions between three communities, which can be thought of as left, center and right parties.

In a stochastic block model, the interaction rate between community $i$ and community $j$ is a constant $a_{ij}$, so that the model can be represented by a matrix $A=\{a_{ij}\}$ or equivalently by a piece-wise constant graphon $W(x,y)=\sum_{i,j}a_{ij}\mathds{1}_{C_i}(x)\mathds{1}_{C_j}(y)$, where $\bigsqcup_i C_i=[0,1]$. In the case considered there are three communities, consequently $C_1=[0,\frac13],$ $C_2=(\frac13,\frac23]$ and $C_3=(\frac23,1]$. 
The idea at the base of the model is to have positive interactions within each community, but negative interactions between left and right parties, while the center mediates and has positive interactions with everyone.  
To account for diversity inside a specific community we change the piece-wise constant graphon into a piece-wise linear graphon. This flexibility allows us, for instance, to distinguish between an extremist (i.e. $x\in [0,\epsilon]$ or $x\in [1-\epsilon,1]$ for a small $\epsilon>0$), which will take into consideration mostly the opinion of its own party, and someone more moderate ($x\in [\frac{1}{3}-\epsilon,\frac{1}{3}]$ or $x\in [\frac{2}{3},\frac{2}{3}+\epsilon]$), which will value the opinion of the center party more, and reject the opposing party's opinions less strongly. Similarly, a member of the centre-left will listen to the left party more than the right party, and vice versa for centre-right. 

Let us consider $C_1$, $C_2$ and $C_3$ respectively as left, center and right party. 
On each set $C_i\times C_j$ we define a plane according to the following rationale:
\begin{enumerate}
    \item for $(x,y)\in C_1\times C_1$, we consider the interaction to be stronger between the extremists, diminishing as $x$ and/or $y$ grow,
    \item for $(x,y)\in C_1\times C_2$, the strongest interaction is between moderate left members and centre-left, and decreases as $x$ becomes more extreme or $y$ grows towards centre-right,
    \item for $(x,y)\in C_1\times C_3$, the interaction is negative, greater in absolute value when both $x$ and $y$ are extremists, and more contained when both are moderate,
    \item for $(x,y)\in C_2\times C_2$, is a constant.
\end{enumerate}
The graphon on the remaining sets can be obtained by symmetry, resulting in the graphon in Fig.~\ref{W_plot}.
For our simulations, we will simulate the dynamics on this graphon and on graphs sampled from it with deterministic latent variables.

The results for the repelling dynamics are shown in Fig.~\ref{fig:diverging-solutions}. We can see that, for the chosen parameters, the three communities' opinions diverge and communities clearly separate. However, when this happens, the opinions within each community do not collapse on the community barycenter, like it happens for the piece-wise constant case (see \cite[Sect. 3.1]{aletti2022opinion} for a detailed example): instead, opinions are more spread out, covering a wider range of values. 

The results for the  opposing dynamics are shown in Fig.~\ref{fig:opposing}. Consistently with well-known results for finite graphs~\cite{altafini2012consensus}, the dynamics converges to 0 because the graphon that we have chosen describes a social network that is not structurally balanced: in fact, there are negative interactions between left and right, but the center has positive interactions with everyone.

Another aspect worth pointing out is the difference between the solutions of the dynamics with graphs and graphons, shown in Fig.~\ref{fig:err_L2}. Consistently with the theoretical results, the error decreases as $n$ grows for both dynamics. As per the dependence on $t$, the bound from Theorem~\ref{thm:bound} allows for the error to grow with time: this growth can indeed be observed for the repelling dynamics and, at least during a transient period, also for the opposing dynamics.

\section{Conclusion}\label{sect:conclusion}

In this paper we considered two models of opinion dynamics with antagonistic interactions (the repelling and the opposing model), and we extended them to graphons. We have shown existence and uniqueness of solutions and stated bounds on the error between graph  and graphon solutions. 
We then showed that these bounds go to zero if we take a sequence of graphs sampled from a graphon, as the number of nodes goes to infinity.

These results prove graphons to be a helpful tool to handle large networks  with antagonistic interactions, as we have shown how dynamics on the graphon well approximate their finite dimension respective versions. Future work is aimed at understanding the general properties of the graphon dynamics itself, such as its steady states and behaviour as time goes to infinity.

\bibliographystyle{IEEEtran}
\bibliography{graphonbiblio}

\end{document}